\documentclass[11pt]{article}
\usepackage{style}
\usepackage{shortcuts}

\title{Preprocessed 3SUM: From Reals to Integers}
\author{
  Nick Fischer\footnote{Max-Planck Institute for Informatics. \href{mailto:nfischer@mpi-inf.mpg.de}{nfischer@mpi-inf.mpg.de}} \and
  Jonas Schmidt\footnote{Bocconi University. \href{mailto:jonas.schmidt2@phd.unibocconi.it}{jonas.schmidt2@phd.unibocconi.it} }
}
\date{}

\begin{document}

\maketitle

\begin{abstract}
\noindent
The \emph{Preprocessed 3SUM} problem asks to preprocess a set $A$ of $n$ numbers so that we can quickly decide whether a query subset $A' \subseteq A$ contains a solution $a+b=c$. The best known query bounds differ substantially between integer and real inputs. For integers, a series of advances with near-quadratic preprocessing reduced the query time from $\tilde O(n^{13/7})$ [\hyperlink{cite.ChanL15}{Chan, Lewenstein; STOC '15}] to $\tilde O(n^{11/6})$ [\hyperlink{cite.ChanWX23}{Chan, Vassilevska W., Xu; STOC '23}], and finally to $\tilde O(n^{3/2})$ [\hyperlink{cite.KasliwalPS25}{Kasliwal, Polak, Sharma; SOSA '25}]. For reals, only the first of these bounds has been achieved, by adapting the corresponding integer algorithm [\hyperlink{cite.Fischer25}{Fischer; SODA '25}]. Is this gap inherent to the more limited tools available for real inputs?

We give a randomized real-to-integer reduction with $\tilde O(n^2)$ preprocessing overhead and no asymptotic increase in query time. This yields $\tilde O(n^2)$ preprocessing and $\tilde O(n^{3/2})$ query time for real inputs, closing the gap. We show how to extend our result to the setting where one set is unknown at preprocessing time. Our reduction is conceptually simple, and relies on an interesting technical toolkit around old and new tools from sparse linear algebra. More broadly, it suggests a route to real-to-integer reductions for other additive problems.
\end{abstract}

\section{Introduction}
We revisit an ever-present question in algorithm design: \emph{``Are problems harder to solve on real-valued inputs than on integer-valued inputs?''} Integer inputs enable techniques such as scaling, modular hashing, and bit manipulation that have no immediate counterparts for arbitrary reals. Understanding whether these advantages are essential has repeatedly led to new algorithms and, sometimes, to entirely new ways of approaching a problem.

A striking recent example is single-source shortest paths with negative edge weights. In 2022, Bernstein, Nanongkai, and Wulff-Nilsen~\cite{BernsteinNW21} gave a breakthrough near-linear-time algorithm for integer weights. However, the algorithm relies on a scaling technique, which seems very hard (if not impossible) to extend to real weights. This left room for a different line of work, initiated by Fineman's breakthrough in 2024~\cite{Fineman24}, which has now culminated in a very different almost-linear-time algorithm for real weights due to Hair, Li, Li, and Zhang~\cite{HairLLZ26}. The effort to bridge the integer--real divide has thus helped drive one of the major recent successes in graph algorithms.

But there are many more interesting examples. Another one is computing \emph{sumsets}\footnote{I.e., given integer sets $A, B$, the goal is to compute the set $A+B=\{a+b:a\in A,\ b\in B\}$.}: for integer sets, near-linear output-sensitive algorithms have been available since 2002~\cite{ColeH02}, but obtaining an analogous algorithm for real sets took more time and a very different set of techniques~\cite{Fischer25}. And of course this question frequently arises in computational geometry; see e.g.~\cite{ChanP09}.

Perhaps the cleanest approach towards approaching such questions is by attempting \emph{real-to-integer} reductions, i.e., reductions from a problem on real inputs to the same problem on integer inputs. This may appear an unrealistic goal at first, and yet, Chan, Vassilevska~W., and Xu~\cite{ChanWX22} have initiated this trend recently in fine-grained complexity by establishing a non-tight reduction from 3SUM on real inputs to 3SUM on integer inputs. Fischer, Polak, and Schmidt~\cite{FischerPS26} recently established the tight equivalence of the integer and real versions of the 3SUM Hypothesis, through an unexpected application of the Frank--Tardos theorem~\cite{FrankT87}. That result has since been overtaken by a dramatic breakthrough~\cite{AlmanV26}, which refuted the 3SUM Hypothesis for both integer and real inputs. The broader question remains: when can we replace real inputs by integers while preserving precisely the structure that an algorithm needs? Its importance extends beyond a particular hardness hypothesis or a particular problem.

\paragraph{Preprocessed 3SUM}
In this work, we consider a variant of 3SUM that has not been affected as dramatically by the recent breakthrough: \emph{Preprocessed 3SUM}. Given a universe $A$ of $n$ numbers, we preprocess~$A$ so that, upon receiving a subset $A'\subseteq A$, we can quickly decide whether there exist $a,b,c\in A'$ with $a+b=c$. The preprocessing is performed once and can be reused for many queries. Equivalently, one can work with three preprocessed universes $A,B,C$ and queries specifying subsets $A'\subseteq A$, $B'\subseteq B$, and $C'\subseteq C$; we use the single-universe formulation for simplicity. 

This problem has an exciting history which we describe in detail soon.
In a nutshell, the state-of-the-art for integers due to Kasliwal, Polak, and Sharma~\cite{KasliwalPS25} is an algorithm with preprocessing time $\tilde O(n^2)$ and query time $\tilde O(n^{3/2})$.
For real inputs, the best known bound is $\tilde O(n^{13/7})$, due to Fischer~\cite{Fischer25}.
The faster integer algorithm uses an integer-specific form of hashing, which makes it difficult to replicate for real numbers.
In fact, Kasliwal, Polak, and Sharma explicitly ask whether their query time can also be achieved over the reals.\footnote{Specifically, they write: \emph{``It seems to remain a difficult open question whether the subsequent improvement to time $\tilde O(n^{11/6})$~\cite{ChanWX23} or our current improvement to time $O(n^{1.5} \log n)$ also can be mirrored in the real RAM.''}}
We discuss later how the recent improvement for ordinary 3SUM can lead to a slightly subquadratic preprocessing time. Nonetheless, it does not seem to affect any of these query times at all and leaves a substantial polynomial gap.

\paragraph{Main result.}
We give a reduction from real to integer Preprocessed 3SUM that preserves the query time and incurs only a near-quadratic additive cost in preprocessing.

\begin{theorem} \label{thm:reduct-real-to-int}
If Preprocessed 3SUM with integer inputs (over the word RAM model) can be solved in preprocessing time $P(n)$ and query time $Q(n)$, then Preprocessed 3SUM with real inputs (over the real RAM model) can be solved in preprocessing time $\tilde O(P(n) + n^2)$ and query time $O(Q(n))$. 
\end{theorem}

Up to the near-quadratic overhead in preprocessing, this is exactly the kind of real-to-integer reduction that we are generally aiming for. Our reduction takes a surprisingly clean form, where we actually replace each real by an appropriate integer and then run the integer algorithm as a black box. It is conceptually much simpler than the recent reduction of Fischer, Polak, and Schmidt~\cite{FischerPS26}, and does not invoke overly complicated machinery like the Frank--Tardos algorithm. While the approach is pleasingly simple on a conceptual level, some interesting technical aspects remain: the algorithm relies on a combination of sparse linear algebra tools, some of which were known~\cite{Wiedemann86,KaltofenS91}, and some of which we develop here.

Applying our reduction to the integer algorithm of Kasliwal, Polak, and Sharma~\cite{KasliwalPS25} answers their question and gives the following bound.

\begin{corollary} \label{cor:real-1.5}
Preprocessed 3SUM with real inputs (over the real RAM model) can be solved in preprocessing time $\tilde O(n^2)$ and query time $\tilde O(n^{3/2})$.
\end{corollary}

This improves the query exponent for real inputs from $13/7$ to $3/2$, matching the best known integer bound with near-quadratic preprocessing. More broadly, and perhaps even more importantly, our method suggests a route to real-to-integer reductions for other additive problems.

\begin{table}[tbp]
    \centering
    \renewcommand{\arraystretch}{1.25}
    \setlength{\tabcolsep}{4pt}
    \begin{tabular*}{\linewidth}{@{\extracolsep{\fill}}lccc@{}}
        \toprule
        \textbf{Source} & \textbf{Preprocessing time} & \textbf{Query time} & \textbf{For reals?} \\
        \midrule
        Chan--Lewenstein~\cite{ChanL15}
            & $\widetilde O(n^2)$ & $\widetilde O(n^{13/7})$ & $\times$ \\
        Chan--Vassilevska W.--Xu~\cite{ChanWX23}
            & $\widetilde O(n^2)$ & $\widetilde O(n^{11/6})$ & $\times$ \\
        Fischer~\cite{Fischer25}
            & $\widetilde O(n^2)$ & $\widetilde O(n^{13/7})$ & $\checkmark$ \\
        Kasliwal--Polak--Sharma~\cite{KasliwalPS25}
            & $\widetilde O(n^2)$ & $\widetilde O(n^{3/2})$ & $\times$ \\
        \midrule
        {\bfseries\cref{cor:real-1.5}}
            & $\widetilde O(n^2)$ & $\widetilde O(n^{3/2})$ & $\checkmark$ \\
        {\cref{thm:preprocessed-subquadratic}}
            & $O(n^{2-\epsilon})$ & $\widetilde O(n^{3/2})$ & $\times$ \\
        \bottomrule
    \end{tabular*}
    \caption{Known algorithms for Preprocessed 3SUM. Here, $\epsilon>0$ is an absolute constant. The table omits recent improvements regarding space complexity~\cite{KirkpatrickKMW26,CarmelKK26}.}
    \label{tab:preprocessed-3sum}
\end{table}

\paragraph{History and the effect of the 3SUM breakthrough.}
Preprocessed 3SUM has a rich history. The problem was first raised in the work of Bansal and Williams~\cite{BansalW12}, who credited Avrim Blum for suggesting this application of preprocessing. Chan and Lewenstein~\cite{ChanL15} gave the first algorithm with truly subquadratic query time $\tilde O(n^{13/7})$ and near-quadratic preprocessing time for integer inputs. Their approach is based on a pioneering application of additive combinatorics. The query time was improved to $\tilde O(n^{11/6})$ by Chan, Vassilevska W., and Xu~\cite{ChanWX23}, and then to $\tilde O(n^{3/2})$ by Kasliwal, Polak, and Sharma~\cite{KasliwalPS25}. None of these results apply to real inputs out of the box; Fischer~\cite{Fischer25} later extended the Chan--Lewenstein algorithm to real inputs substituting the integer-specific sparse convolution routine with an analogue for reals. \Cref{tab:preprocessed-3sum} summarizes the relevant bounds.
More recently, Kirkpatrick, Kuszmaul, Mathialagan, and Vassilevska W.~\cite{KirkpatrickKMW26}, and Carmel, Kosoburd, and Krauthgamer~\cite{CarmelKK26}, studied space-efficient data structures for a variant of Preprocessed 3SUM where the universe set is partially unknown. In this variant, only two sets $A, B$ are preprocessed, and each query is to decide whether $A' \subseteq A, B' \subseteq B$ and $C' \subseteq \Int$ contain a 3SUM solution.
A related line of work studies a different setting of 3SUM that allows for preprocessing, namely \emph{3SUM-Indexing}~\cite{GoldsteinKLP17,KopelowitzP19,GolovnevGHPV20,ChungL23,DinurG26}, with the goal to find the best tradeoff between space and query time for a data structure that preprocesses integer sets $A, B$ and answers queries of the form: given $c \in \Int$, is there $(a,b) \in A \times B$ such that $a + b = c$?

The new algorithm for ordinary 3SUM also has a consequence for Preprocessed 3SUM, although obtaining it requires more than invoking that algorithm at query time. In the appendix, we show how to combine the breakthrough~\cite{AlmanV26} with the existing preprocessing techniques to obtain an integer algorithm with preprocessing time $O(n^{2-\epsilon})$ and query time $\tilde O(n^{1.5})$, for some absolute constant $\epsilon>0$; see \cref{thm:preprocessed-subquadratic}. The saving in the preprocessing exponent is small, but establishes that truly subquadratic preprocessing is compatible with the best known query bound.

\paragraph{Other settings.}
Algorithms for preprocessed 3SUM have been successfully extended to several other directions:
(i) derandomizing the algorithms, (ii) balancing query time and space usage, and (iii) allowing the third query set to be unknown at preprocessing time.
Perhaps the technically most interesting extension for our algorithm is the third setting, where sets $A$ and $B$ are given at preprocessing time and queries consist of $A' \subseteq A$, $B' \subseteq B$, and an arbitrary third set $C'$. Algorithms must detect a solution $(a,b,c) \in A' \times B' \times C'$ to $a + b = c$.
We resolve this setting with a similarly clean reduction as \Cref{thm:reduct-real-to-int}, leading to the following corollary. See \Cref{sec:unknownC} for how to adapt our algorithm to also solve this setting.

\begin{restatable*}{corollary}{corunknownC}\label{cor:unknown-c}
    Preprocessed 3SUM with Unknown $C$ with real inputs (over the real RAM model) can be solved in preprocessing time $\tilde O(n^2)$ and query time $\tilde O(n^{3/2})$.
\end{restatable*}

\paragraph{The real RAM model.}
The real RAM is a convenient model for algorithms on arbitrary real inputs, but its unrestricted form can be unrealistically powerful. In particular, unit-cost arithmetic with unrestricted multiplication can conceal computations on numbers of enormous bit complexity. See, e.g.,\ the discussion in Fischer, Polak, and Schmidt~\cite{FischerPS26}. Notably, our reduction works in a much more restricted model which only supports queries of the form ``$a \leq b$?'' and ``$a + b \leq c$?'' (and reverse inequalities) for input reals $a, b, c$. In particular, it does not require storing results of arithmetic operations. In this restricted model, $\Omega(n^2)$ preprocessing time is \emph{unconditionally} necessary to support truly subquadratic query time~\cite{Erickson95,AilonC05}, making our preprocessing bound \emph{optimal} up to polylogarithmic factors.

\section{Preliminaries}

We formally define the problem.

\begin{definition}[Preprocessed 3SUM]\label{def:preproc-3sum}
Preprocess a given set of numbers $A$ of size $n$ to support the following queries: Given $A' \subseteq A$, decide if $A'$ contains a solution $a + b = c$.
\end{definition}

This definition slightly differs from the usual, \emph{tripartite} definition of preprocessed 3SUM, where we preprocess three sets $A$, $B$, and $C$ to support queries of the form: Given $A' \subseteq A$, $B' \subseteq B$, and $C' \subseteq C$, decide if there is a solution $(a,b,c) \in A' \times B' \times C'$ to $a + b = c$. 
It is well known that for ordinary 3SUM, the tripartite version and the single-set version are equivalent in running time up to logarithmic factors via simple reductions (see, e.g.,~\cite{DudekGS20}).
These reductions immediately carry over to the preprocessed setting.
Therefore, we can choose to work with \Cref{def:preproc-3sum}, which is more convenient for us.

We consider Preprocessed 3SUM with either real or integer inputs.
Unless specified otherwise, we assume that integer inputs are bounded by $\set{-n^{O(1)}, \dots, n^{O(1)}}$.
All randomized algorithms are Monte Carlo randomized. We sometimes write \emph{with high probability} to say that an event happens with probability at least $1 - n^{-c}$ for some sufficiently large constant $c$.

Let $A \in \Field_p^{m \times n}$ be a matrix over the finite field $\Field_p$ for a prime $p$. We denote by $\nnz(A) \coloneqq \abs{\set{(i,j)\in[m]\times[n]:A[i,j] \ne 0}}$ the number of nonzero entries of $A$. Similarly, for a vector $v \in \Field_p^n$, the number of nonzero entries in $v$ is denoted by $\nnz(v)$.
We will equivalently write that $v$ is \emph{$t$-sparse} if $\nnz(v) \le t$.

\section{Real-to-Integer Reduction}
We start with a high-level description of our reduction. Our starting point is the same as that of Fischer, Polak, and Schmidt~\cite{FischerPS26}. In their work, they design a real-to-integer reduction for 3SUM. That is, they show that if integer 3SUM can be solved in time $O(n^{2-\epsilon})$ then real 3SUM can be solved in time $O(n^{2-\epsilon'})$. The dependence of $\epsilon'$ on $\epsilon$ is quite bad, and there is no hope that their approach (as is) can be used to lead to a lossless reduction. (And again, the result has since been subsumed by~\cite{AlmanV26}.) Nevertheless, the same basic observations are useful. Their result is based on the following theorem due to Frank and Tardos~\cite{FrankT87} (see also~\cite{Eisenbrand10}):

\begin{theorem}[Frank--Tardos]\label{thm:ft}
Fix any integer $c \geq 1$. There is a deterministic algorithm that, given reals $a = (a_1, \dots, a_n) \in \Real^n$, computes integers $v = (v_1, \dots, v_n) \in \Int^n$ satisfying the following two properties:
\begin{itemize}
    \item $a^T x \geq 0$ if and only if $v^T x \geq 0$, for all $x \in \set{-c, \dots, c}^n$,
    \item $\max_i |v_i| \leq \exp(O(n^{3}))$.
\end{itemize}
The algorithm runs in polynomial time $n^{O(1)}$.  
\end{theorem}

In words, this theorem allows to replace arbitrary reals by integers such that the order of all subset sums is preserved. Moreover, the integers are reasonably small with bit representations of polynomial length $O(n^3)$. This should be viewed as a ``big hammer''. It can be directly applied to 3SUM to yield the intended real-to-integer reduction, albeit with a (large) polynomial overhead in the preprocessing time. 

Specifically, apply the Frank--Tardos theorem to the given set of reals $A = \set{a_1, \dots, a_n}$. Then substitute each real $a_i$ by the corresponding integer $v_i$. This preserves exactly all 3SUM solutions and non-solutions. If $a_i + a_j = a_k$ then also $v_i + v_j = v_k$ (by applying the first condition two times, for the vectors $x = e_i + e_j - e_k$ and $x = -e_i - e_j + e_k$).

The only remaining problem is that the resulting integers are very large; however, this can be solved by a standard trick. We can hash all integers $v_1, \dots, v_n$ to the range $\set{-n^{O(1)}, \dots, n^{O(1)}}$ in such a way that 3SUM solutions remain preserved. A simple way is to sample a random prime~$p$ of size $n^{O(1)}$ and to reduce all numbers modulo $p$. In this way, 3SUM solutions are necessarily preserved, while non-solutions stay non-solutions with good probability.

Except for this last technical issue, the approach can be summarized by the following observation:

\begin{observation}\label{obs:ai}
There is a randomized polynomial-time algorithm which, given a set $A$ of $n$ reals, computes an \emph{integer encoding} $f : A \to \Int$, i.e., a function satisfying, for all $a, b, c \in A$, that $a + b = c$ if and only if $f(a) + f(b) = f(c)$.
\end{observation}

However, applying the Frank--Tardos theorem appears to be a big hammer for a much smaller task. It not only preserves 3SUMs, but even Subset Sums and even the \emph{order} of Subset Sums.

It turns out that this observation has a simpler alternative proof. (This alternative proof is exactly what was pointed out by an LLM at an early stage of this project, see \hyperref[sec:ai]{AI Disclosure}.) Write \makebox{$A = \set{a_1, \dots, a_n}$}. Consider the matrix $M$, where we insert the row $e_i + e_j - e_k$ if and only if $a_i + a_j = a_k$ is a 3SUM solution. Note that every integer encoding $f$ of $A$ forms a solution to the equation system $M x = 0$ (when taking $x_i = f(a_i)$). In this spirit, the idea is to compute a \emph{random}\footnote{Of course, it is not well-defined here what a random integer solution means. We will later make this rigorous by working over a finite field $\Field_p$ instead.} integer solution to the equation $M x = 0$ (i.e., a random kernel element of $M$) and to take $f(a_i) \coloneqq x_i$.

Clearly, the resulting function $f$ preserves all 3SUM solutions. More interestingly, focus on any non-solution $a_i + a_j \neq a_k$. The constraint $x_i + x_j = x_k$ is not implied by the equation system. Thus, the set of solutions $x$ satisfying that extra constraint is of dimension one smaller than the dimension of the kernel. For this reason, the probability of sampling one such solution is small, and with high probability the resulting function $f$ satisfies $f(a_i) + f(a_j) \neq f(a_k)$.

The problem, of course, is to implement this alternative proof more economically than the large polynomial time of the naive algorithm. Currently the running time suffers from large polynomial factors for three different reasons:
\begin{enumerate}
    \item Listing the up to $n^2$ 3SUM solutions in time $O(n^2)$.
    \item Solving an equation system of size up to $n^2 \times n$. This takes time $O(n^4)$ naively, and fast matrix multiplication allows some improvements to time $O(n^3)$ at best. To achieve time $\tilde O(n^2)$ we have to exploit the \emph{sparsity} of the system.
    \item Any nonzero integer solution to this equation system must have polynomially many bits in the worst case. E.g., consider the 3SUM instance $\set{1, 2, 4, 8, \dots, 2^n}$; this set cannot be replaced with integers smaller than $2^n$ without introducing or destroying some solutions $a + b = c$. Thus, each basic computation with a resulting integer already requires time $\Omega(n)$.
\end{enumerate}

While 1 is acceptable (and unavoidable in the restricted real RAM model), 2 and 3 are real issues. We will deal with 2 by deriving a new algorithm for solving rectangular, row-sparse equation systems; see the following \cref{thm:sparse-rect-eq-sys}. This algorithm is likely to be of independent interest. To deal with~3, as before, we instead perform all computations over sufficiently large finite fields.

\begin{restatable*}[Sparse Rectangular Linear Equation Solving]{theorem}{thmsparserecteqsys} \label{thm:sparse-rect-eq-sys}
There is a randomized algorithm that, given a matrix \smash{$A \in \Field_p^{m \times n}$} and a vector \smash{$b \in \Field_p^m$}, returns a uniformly random solution $x$ to $A x = b$ (or reports that no solution exists). The algorithm runs in time \smash{$\tilde O((m + n^2) t)$}, provided that each row in $A$ is $t$-sparse and that $p > m^4$.
\end{restatable*}

For us, the parameter regime of interest is $m \approx n^2$ and $t = O(1)$. In this regime the algorithm takes time $\tilde O(n^2)$, which is \emph{near-optimal}, since already specifying the input takes time $\Omega(n^2)$.

\subsection{Technical Proof}

\begin{proof}[Proof of \cref{thm:reduct-real-to-int}]
We start with a description of the preprocessing phase. Given $A \subseteq \Real$, we first list all 3SUM solutions in time $O(n^2 \log n)$. That is, we first sort $A = \set{a_1 < \dots < a_n}$ in time $O(n \log n)$. Then, for each pair $i, j \in [n]$ we binary-search whether $a_i + a_j \in A$. (Note that this algorithm can indeed be implemented in the restricted real RAM model mentioned before.) Let $S = \set{(i, j, k) \in [n]^3 : a_i + a_j = a_k}$ denote the set of solutions.

Sample a uniformly random prime $p$ in the range $[n^{10}, 2n^{10}]$, and consider the matrix $M$ of size $|S| \times n$ over $\Field_p$, where the row corresponding to $(i, j, k) \in S$ is exactly $e_i + e_j - e_k$. Compute a random solution $x \in \Field_p^n$ to $M x = 0$ by one call to \cref{thm:sparse-rect-eq-sys}. Below we will prove that, with high probability, $a_i + a_j = a_k$ if and only if $x_i + x_j = x_k$. By that claim, it suffices to solve the $\Field_p$ version of Preprocessed 3SUM: preprocess the set $X = \set{x_1, \dots, x_n} \subseteq \Field_p$, and then upon query of a subset $X' \subseteq X$, decide whether $X'$ contains a 3SUM solution (modulo $p$). This problem can easily be reduced to 2 instances of integer Preprocessed 3SUM over the universe $\set{0, \dots, O(p)}$ (one for the solutions with $x_i + x_j \in \set{0, \dots, p-1}$ and one for the solutions with $x_i + x_j \in \set{p, \dots, 2p-2}$, when viewing the $x_i$'s as integers).

\begin{claim}
With high probability, $a_i + a_j = a_k$ if and only if $x_i + x_j = x_k$ for all $i, j, k \in [n]$.
\end{claim}
\begin{proof}
Apply the Frank--Tardos theorem (\Cref{thm:ft}) on $a = (a_1, \dots, a_n) \in \Real^n$ to obtain integers $v \in \Int^n$ bounded by $\norm{v}_\infty \leq \exp(O(n^3))$ such that $a_i + a_j = a_k$ if and only if $v_i + v_j = v_k$. We claim that, with high probability, this is further equivalent to $v_i + v_j \equiv v_k \mod{\,p}$. Indeed, the only failure event is when $v_i + v_j \neq v_k$ and $p$ divides $v_i + v_j - v_k$. This number has at most $\log (2 \norm{v}_\infty) \leq O(n^3)$ prime divisors, thus the probability of selecting $p \in [n^{10}, 2n^{10}]$ among these $O(n^3)$ bad choices is at most $\tilde O(n^{-7})$ by the Prime Number Theorem. By a union bound over the $n^3$ choices of $(i, j, k)$, the total error probability is $\tilde O(n^{-4})$.

Thus, view $v$ as an element of $\Field_p^n$. By construction we have $M v = 0$ (now viewed over $\Field_p$). Sampling a uniformly random kernel element $x \in \ker M$ can therefore be equivalently viewed as sampling $y \in \ker M$ and a uniformly random scalar $\alpha \in \Field_p$, and setting $x \coloneqq y + \alpha v$. But then, for any triple $(i, j, k)$ with $a_i + a_j \neq a_k$ we have that $x_i + x_j - x_k = (y_i + y_j - y_k) + \alpha (v_i + v_j - v_k)$. Since $v_i + v_j - v_k \neq 0$ and since $\alpha \in \Field_p$ is sampled uniformly and independently, it follows that also $x_i + x_j - x_k$ is uniformly distributed in $\Field_p$. The probability of hitting zero is thus $1/p$, and by a union bound over all non-solution triples the total error probability is $O(n^{-7})$. 
\end{proof}

Let us finally analyze the running time. Listing the 3SUM solutions takes time $\tilde O(n^2)$. The call to \cref{thm:sparse-rect-eq-sys} is with parameters $m = O(n^2)$ and $t = 3$, and thus takes time $\tilde O(n^2)$, too. Finally, we preprocess two instances of integer Preprocessed 3SUM in time $O(P(n))$. In the query phase, we issue two queries to respective two instances taking time $O(Q(n))$ as claimed.
\end{proof}
\section{Sparse Linear Algebra}
It remains to design the fast algorithm for solving sparse, rectangular equation systems.

\thmsparserecteqsys

Throughout this section, $p$ is a prime and we work over the finite field $\Field_p$. We assume that $p$ fits into $O(1)$ machine words so that field operations in $\Field_p$ can be executed in $O(1)$ time. 

\subsection{Technical Overview}

We provide a quick overview of the main ideas behind \Cref{thm:sparse-rect-eq-sys}.
In our setting, think of $A$ as an $n^2 \times n$ matrix with only three nonzero entries per row, encoding all 3SUM solutions.

Technically, our algorithm can be viewed as a reduction to solving a sparse linear equation $Ax = b$ where $A$ is a square, $n \times n$ matrix with again only $O(1)$ entries per row.
For this setting, there already exists an efficient algorithm due to Wiedemann~\cite{Wiedemann86}, which would yield time $\tilde O(n^2)$ with our parameters:

\begin{lemma}[Sparse Linear Equation Solving~\cite{Wiedemann86}] \label{lem:sparse-eq-sys}
There is a randomized algorithm that, given a matrix \smash{$A \in \Field_p^{m \times n}$} and a vector \smash{$b \in \Field_p^m$}, returns a uniformly random solution $x$ to $A x = b$ (or reports that no such solution exists) in time $\tilde O((n + m) \nnz(A))$.
\end{lemma}

In fact, the same paper gives algorithms for the rectangular setting, however for our parameters the running time bound becomes $\tilde O(n^3)$, which is insufficient.
The main bottleneck in Wiedemann's algorithm for rectangular matrices as well as follow-up improvements~\cite{KaltofenS91}, is that they solve the rectangular case by first (implicitly) computing a basis of the row space. 
At this point, the algorithm only has to deal with $\rank(A)$ many rows and we are basically back at the fast, square case from \Cref{lem:sparse-eq-sys}.
However, even the state-of-the-art approach to compute the rank and such a linearly independent set of columns runs in time $O(\nnz(A) + \rank(A)^\omega)$, or $O(n^\omega)$\footnote{Here, $\omega$ denotes the running time exponent of fast matrix multiplication.} in our setting~\cite{CheungKL12}. 
Hence, these approaches cannot yield near-quadratic running time if $\omega > 2$.

Our idea is to replace this expensive reduction to the square case by a simpler one.
Indeed, instead of computing a basis of the row space, it is sufficient to find $O(n)$ rows that span the full row space, since then Wiedemann's algorithm (\Cref{lem:sparse-eq-sys}) solves the equation in time $\tilde O(n^2)$.
We show that such an \emph{approximate basis} can be computed efficiently with a relatively simple, randomized algorithm.
The key insight behind this algorithm is the following win/win scheme:
If many linearly independent rows are easy to detect because most subsets of rows have a high rank, then a small random subset of rows will also have a high rank.
On the other hand, if very specific rows need to be selected to span the full row span, then a random subset of rows allows us to detect and remove many obsolete linearly dependent rows (using weaker but more efficient sparse linear algebra tools).
Recursively applying the same idea yields an approximate basis after $O(\log n)$ iterations.
See \Cref{lem:sparse-approx-basis} in \Cref{sec:rectangular} for a formal description of this algorithm.

\subsection{Basic Sparse Linear Algebra}

We rely on a basic computational fact due to Kaltofen and Saunders~\cite{KaltofenS91}.
\begin{lemma}[Sparse Rank Computation~\cite{KaltofenS91}] \label{lem:sparse-rank}
There is a randomized algorithm that, given a matrix \smash{$A \in \Field_p^{m \times n}$}, computes $\rank(A)$ with high probability in time $\tilde O((n + m) \nnz(A))$.
\end{lemma}

We will also rely on the following lemma which is implicit in~\cite{Wiedemann86}; we provide a quick (and somewhat simplified) proof for convenience.

\begin{lemma}[\cite{Wiedemann86}] \label{lem:random-sparse-matrix-det}
Let $A \in \Field_p^{m \times n}$ be a matrix with rank at most $r$, and let $\delta > 0$. Let $H \in \Field_p^{n \times r}$ be a matrix where each entry is sampled independently such that with probability $q = \min\set{10 (\log n + \log \delta^{-1}) / r, 1}$ the entry is uniformly random from $\Field_p$, and with the remaining probability $1 - q$ the entry is zero. Then,
\begin{equation*}
    \Pr(\rank(A H) = \rank(A)) \geq 1 - O\parens*{\frac{n}{p}} - O(\delta).
\end{equation*}      
\end{lemma}
\begin{proof}
It always holds that $\rank(A H) \leq \rank(A)$, so focus on the other direction. Put $s := \rank(A) \leq r$. Without loss of generality we can assume $m=s$ by removing all linearly dependent rows from $A$. The goal is to show that the $s\times r$ matrix $A H$ contains a nonsingular $s\times s$ minor.

To show this, we view the process of sampling $H$ in two stages. In the first stage, we sample a \emph{mask} matrix $M \in \set{0, {*}}^{n \times r}$, where each entry is ${*}$ with probability~$q$. In the second stage we replace ${*}$ by uniformly random field elements to obtain $H$.

We start analyzing the first stage. Let $J \subseteq [n]$ be a maximum subset of linearly independent columns in~$A$, so $|J| = s$. Let $A_J$ denote the matrix restricted to these columns (so that \makebox{$\det(A_J) \neq 0$}), and let $M_J$ denote the submatrix of $M$ restricted to the rows in $J$. Then $M_J$ can be viewed as the adjacency matrix of a bipartite graph with $s$ vertices on the left and $r$ vertices on the right, where each edge is sampled independently with probability~$q$. With probability $1-O(\delta)$, that graph contains a matching covering all $s$ vertices on the left~\cite{ErdosR63}.\footnote{By Hall's theorem, a matching covering the left side fails to exist only if there is a set $L$ of $k$ left vertices whose neighbors are contained in a set of $k-1$ right vertices. A union bound gives failure probability at most
\begin{equation*}
    \sum_{k=1}^s \binom{s}{k}\binom{r}{k-1} e^{-qk(r-k+1)}
    \leq \sum_{k=1}^r \binom{r}{k}\binom{r}{k-1} e^{-qk(r-k+1)}
    \leq 2\sum_{k=1}^{\ceil{r/2}} n^{2k} e^{-qkr/2}
    = O(\delta).
\end{equation*}
Here we used that $q < 1$ and thus $qr=10(\log n+\log\delta^{-1})$, $r \leq n$, and $\delta\leq1/2$. If $q = 1$, the statement is trivially satisfied.} Fix such a matching and let $I \subseteq[r]$ be its set of $s$ matched right vertices. This choice depends only on the mask.

Now focus on the second stage when we replace each ${*}$ by a uniformly random field element to obtain~$H$. For the moment, pretend that each random field element is instead a symbolic variable. Let $x$ denote the set of these variables, and write $H(x)$ to denote the resulting symbolic matrix. Let $H_I(x)$ be its restriction to the columns in $I$, and let $H_{J, I}(x)$ be its restriction to the rows in $J$ and columns in $I$. By the previous paragraph, the determinant $\det(H_{J, I}(x))$ is not identically zero (viewed as a polynomial in $x$), as the chosen matching corresponds to a nonzero term in $\det(H_{J, I}(x))$. The Cauchy--Binet formula further implies that
\begin{equation*}
    \det(A H_I(x)) = \sum_{\substack{K \subseteq [n]\\|K| = s}} \det(A_K) \det(H_{K, I}(x))
\end{equation*}
is not identically zero. Indeed, the monomial corresponding to the chosen matching has coefficient $\pm\det(A_J)\neq0$ and cannot occur in any other summand, since its variables use exactly the rows in~$J$.

In other words, $\det(A H_I(x))$ is a nonzero polynomial with degree at most~$s$. Thus, the Schwartz--Zippel lemma implies that substituting random field elements leads to $\det(A H_I(x))=0$ with probability at most $s/p \leq n/p$. Otherwise $A H$ has a nonsingular $s \times s$ minor and hence rank $s$. Together with the first-stage failure probability, this proves the claimed bound.
\end{proof}

\begin{lemma}[Linear Dependence Testing] \label{lem:lin-dep-testing}
Let $\delta > 0$. There is a randomized algorithm that preprocesses (nonzero) vectors $a_1, \dots, a_m \in \Field_p^n$ in time $\tilde O((n + m) \sum_i \nnz(a_i) \cdot \log \delta^{-1})$, and then supports the following queries: given $v \in \Field_p^n$, decide whether $v$ is linearly dependent from $a_1, \dots, a_m$ in time $O(\nnz(v) (\log n + \log \delta^{-1}))$. Each query is answered correctly with probability $1 - O(n / p) - \delta$.
\end{lemma}
\begin{proof}
Let $A \in \Field_p^{m \times n}$ be the matrix with rows $a_1, \dots, a_m$, and compute $\rank(A)$ by \cref{lem:sparse-rank}. We first describe a simple algorithm for the special case that $\rank(A) \geq n - 1$, and then later give a reduction to this case. If $\rank(A) = n$ then all queries are trivial. If $\rank(A) = n - 1$, then the kernel of $A$ has dimension exactly $1$, and thus $\ker A = \langle x\rangle$ for some nonzero \smash{$x \in \Field_p^n$}. We can compute one such $x$ by solving one equation system with \cref{lem:sparse-eq-sys}. The total preprocessing time is $\tilde O((n + m) \nnz(A))$.

Now focus on the query phase. A vector $v$ is linearly dependent from $a_1, \dots, a_m$ if and only if $x^T v = 0$. Indeed, $x^T v \neq 0$ if and only if the matrix with row vectors $a_1, \dots, a_m, v$ has no nontrivial kernel and thus has full rank. Thus, upon query of a vector $v$ it suffices to evaluate the inner product $v^T x$ in time $O(\nnz(v))$.

It remains to reduce to the case $\rank(A) \geq n - 1$. So suppose that $r := \rank(A) < n - 1$. We sample \smash{$H \in \Field_p^{n \times (r+1)}$} as in \cref{lem:random-sparse-matrix-det}. That is, we let $q = \min(10 (\log n + \log \delta^{-1}) / (r + 1), 1)$, and we independently choose each entry to be $0$ with probability $1 - q$ and to be a uniformly random field element with probability $q$. We compute the vectors \smash{$b_1, \dots, b_m \in \Field_p^{(r + 1)}$} defined by $b_i = a_i H$ (viewing $a_i$ as a row vector). Analogously, we write $B$ for the matrix with rows $b_1, \dots, b_m$. Then \cref{lem:random-sparse-matrix-det} implies that $\rank(B) = r$ with high probability, and thus we have $\rank(B) = (r + 1) - 1$, which is as desired as $r + 1$ takes the new role of $n$.

Upon query of a vector $v$, the goal is to check whether $v$ is linearly dependent from $a_1, \dots, a_m$. We instead compute $w := v H$ and query whether $w$ is linearly dependent from $b_1, \dots, b_m$. The claim is that these two conditions are equivalent (with high probability). Indeed, consider the matrix $\hat A$ with row vectors~$a_1, \dots, a_m, v$ and the matrix $\hat B$ with row vectors $b_1, \dots, b_m, w$. Then clearly $\hat B = \hat A H$. Moreover, $v$ is linearly independent from~$a_1, \dots, a_m$ if and only if $\rank(\hat A) = r + 1$. By \cref{lem:random-sparse-matrix-det}, this is further equivalent to $\rank(\hat B) = \rank(\hat A H) = r + 1$, and so the correctness follows. The error probability is $O(n / p) + O(\delta)$ per query. Rescaling $\delta$ by a constant factor gives the stated guarantee.

We finally analyze the running time per query. Each row of $H$ contains at most $O(q r) = O(\log n + \log \delta^{-1})$ nonzero entries with high probability. Thus, with high probability the computation of $w = v H$ takes time $O(\nnz(v) (\log n + \log \delta^{-1}))$, and the resulting vector $w$ satisfies $\nnz(w) \leq O(\nnz(v) (\log n + \log \delta^{-1}))$. Querying $w$ in the data structure described before takes time $O(\nnz(w))$. All in all, the query time is as claimed.
\end{proof}

\subsection{Solving the Rectangular Case}\label{sec:rectangular}

\begin{theorem}[Approximate Basis Computation] \label{lem:sparse-approx-basis}
There is a randomized algorithm that, given $t$-sparse vectors $a_1, \dots, a_m \in \Field_p^n$, selects a subset $S \subseteq [m]$ of size $|S| \leq 16n$ such that each vector in $a_1, \dots, a_m$ is linearly dependent from the vectors $\set{a_i}_{i \in S}$. The algorithm runs in time \smash{$\tilde O((m + n^2) t)$} and succeeds with high probability $1 - 1/m$ provided that $p > m^4$.
\end{theorem}
\begin{proof}
The algorithm is as follows. If $m \leq 16n$, then select all vectors. Otherwise, $m > 16n$. In this case our goal is to \emph{remove} many vectors. Note that a vector can always be safely removed if it is linearly dependent from the other remaining ones.

Specifically, select $T \subseteq [m]$ uniformly at random of size $4n$. Preprocess these vectors by \cref{lem:lin-dep-testing} (with parameter $\delta = m_0^{-10}$, where $m_0$ denotes the original input size). Then scan over all other indices $j \in [m] \setminus T$, and test whether $a_j$ is linearly dependent from $(a_i)_{i \in T}$. If it is linearly dependent, we delete $a_j$. Once this process stops, we recurse on the set of non-deleted vectors.

It is clear that the algorithm is correct (provided it terminates). It remains to argue that the algorithm makes progress. To argue this, analyze the process of sampling $T$ in steps. Initially, $T$ is empty. After each step we select an index from $[m] \setminus T$ uniformly at random and include it into $T$. Before sampling the $k$-th index, let
\begin{equation*}
    N = \abs*{\set*{j \in [m] \setminus T : \text{$a_j$ is linearly independent from $\set{a_i}_{i \in T}$}}}.
\end{equation*}
There are two cases:
\begin{enumerate}[label=(\roman*)]
    \item \emph{If $N > m/2$:} Then with probability more than $1/2$ the rank of the selected vectors $(a_i)_{i \in T}$ increases by including a uniformly random $j \in [m] \setminus T$ into $T$. Let $\tau$ denote the first time that $N \leq m/2$. In each step before $\tau$, the rank of the selected rows increases with (conditional) probability at least $1/2$. Since the rank cannot exceed $n$, we must have $\Ex[\tau] \leq 2n$. By Markov's inequality it follows that $\Pr(\tau \leq 4n) \geq 1/2$. In other words, with probability at least $1/2$ there is a step in which we reach case~(ii):
    \item \emph{If $N \leq m/2$:} Then at least $m - |T| - N \geq m / 2 - |T|$ of the remaining vectors are linearly dependent from $(a_i)_{i \in T}$. Among those, the algorithm may select up to $4n - |T|$, but the others will certainly be deleted. Thus, the number of deleted vectors is at least
    \begin{equation*}
        \frac{m}{2} - |T| - (4n - |T|) \geq \frac{m}{2} - \frac{m}{4} \geq \frac{m}{4}.
    \end{equation*}
\end{enumerate}

Putting both arguments together, the algorithm deletes at least $\Omega(m)$ vectors with probability at least $\Omega(1)$. It follows that the recursion depth is at most $O(\log m)$ (in expectation, and also with high probability by a simple Chernoff bound).

Finally, we analyze the running time and error probability. The set $T$ can be selected in time $O(m)$. Preprocessing the selected vectors by \cref{lem:lin-dep-testing} takes time $\tilde O(n^2 t)$. Each of the $m$ queries takes time $O(t \log m)$, so the total time per recursive call is $\tilde O((m + n^2) t)$. As argued before, the total overhead due to the recursion is at most $O(\log m)$. Moreover, each of the $O(m \log m)$ queries fails with probability at most $O(n / p) + \delta$. By a union bound, the total error probability is thus $\tilde O(n m / p + m^{-9})$.
\end{proof}

\begin{proof}[Proof of \cref{thm:sparse-rect-eq-sys}]
To solve the given equation system $A x = b$, we first run \cref{lem:sparse-approx-basis} to select a submatrix $A'$ of $A$ consisting of $O(n)$ rows but with the same rank. Let $b'$ be the restriction of $b$ to the same set of rows. We then compute a uniformly random solution to the equation system $A' x = b'$ by \cref{lem:sparse-eq-sys}; if no such solution exists we report so. We then test whether $A x = b$. If yes, then the equation system is consistent and we report $x$. If no, then the equation system is inconsistent and there cannot be a solution.

The running time of \cref{lem:sparse-approx-basis} is $\tilde O((m + n^2) t)$ and the running time of \cref{lem:sparse-eq-sys} is $\tilde O(n^2 t)$. All other steps can be implemented in time $O((m + n) t)$.

\end{proof}

\section{Algorithm for Unknown \texorpdfstring{$C$}{C}}\label{sec:unknownC}

We show how the proof of \Cref{thm:reduct-real-to-int} can be adapted to solve the variant of preprocessed 3SUM where one of the three sets is not known during the preprocessing phase.
All known algorithms for integer preprocessed 3SUM also give bounds for this setting~\cite{ChanL15,ChanWX23,KasliwalPS25,KirkpatrickKMW26,CarmelKK26}.
Note that this task does not make sense in the single-set version of 3SUM, so we have to work with multiple sets in this section.
To make it easier to apply our linear algebra toolkit, we work with a slightly odd-looking, two-set problem definition:

\begin{definition}[Preprocessed 3SUM with Unknown $C$]
    Preprocess a given set of numbers $A$ of size $n$, to support the following queries:
    Given $A'\subseteq A$ and any set $C$ of size $n$, decide if there is a solution $(a,a',c) \in A'\times A' \times C$ to $a + a' = c$.
\end{definition}

It is again simple to convince yourself that this two-set definition is equivalent to the standard three-set definition via standard reductions~\cite{DudekGS20}.

In this section, we prove the following reduction.
\begin{theorem}\label{thm:unknown-c}
    If Preprocessed 3SUM with Unknown $C$ with integer inputs (over the word RAM model) can be solved in preprocessing time $P(n)$ and query time $Q(n)$, then Preprocessed 3SUM with Unknown $C$ with real inputs (over the real RAM model\footnote{Note that this algorithm works in a real RAM model that is as restricted as the algorithm of \Cref{thm:reduct-real-to-int}, except for also allowing comparisons involving four input reals, that is, of the form ``$a + b \le c + d$?''.}) can be solved in preprocessing time $\tilde O(P(n) + n^2)$ and query time $\tilde O(Q(n))$.
\end{theorem}

Let us comment on the main issues we face when adapting our algorithm to the setting with unknown $C$ and how we circumvent them.
Recall that our reduction works by simply writing down all 3SUM-solutions that the real instance satisfies, followed by solving a sparse linear equation system of size $n^2 \times n$ with up to $O(n^2)$ nonzero entries.
However, without knowing one of the three sets, we a priori cannot write down any 3SUM-equation that should be satisfied.
Instead, a natural idea is to write down four-term equations $a_1 + a_2 = a_3 + a_4$ over $A$, which are enough to pinpoint all necessary information about $A$ in the instance.
The remaining issue is that the number of these equations can be up to $n^3$ (since three numbers determine the fourth), leading to a blow-up in the size of the equation system.

We resolve this issue by constructing a smaller equation system that still encodes the required additive information as follows: We make sure that each pair $(a_1, a_2) \in A^2$ is involved in at most two equations, leading to again only $O(n^2)$ many rows in our matrix. We choose the equations in such a way that all missing equalities will be implied by transitivity.
Together, with the state-of-the-art algorithm for integers due to Kasliwal, Polak, and Sharma~\cite{KasliwalPS25}, we then arrive at an algorithm for reals with the same guarantees.

\corunknownC
\begin{proof}[Proof of \Cref{thm:unknown-c}]
    Like in the proof of \Cref{thm:reduct-real-to-int}, we can assume without loss of generality that the given algorithm for integer inputs solves the $\Field_p$ version of Preprocessed 3SUM with Unknown $C$ (where $p$ is a prime and solutions are modulo $p$).

    We first describe the preprocessing phase, where we are given a set $A \subseteq \Real$.
    We denote the entries of $A$ by $a_1, \ldots, a_n$ arbitrarily and maintain a set $S \subseteq [n]^4$, which is initially empty.
    We start by sorting the set of all $n^2$ pairs $(a_i, a_j) \in A^2$ by their sum in non-decreasing order in time $O(n^2 \log n)$.
    Then, we consider all $n^2-1$ choices of consecutive pairs $(a_i, a_j)$ and $(a_k, a_\ell)$ in this ordering.
    If $a_i + a_j = a_k + a_\ell$, we insert $(i, j, k, \ell)$ into $S$.
    
    Next, sample a uniformly random prime $p$ in the range $[n^{10}, 2n^{10}]$, and consider the matrix $M$ of size $|S| \times n$ over $\Field_p$, where the row corresponding to $(i, j, k, \ell) \in S$ is now exactly $e_i + e_j - e_k - e_\ell$. Compute a random solution $x \in \Field_p^n$ to $M x = 0$ by one call to \cref{thm:sparse-rect-eq-sys} and give the set of entries of $x$ to the integer algorithm to preprocess.

    Now, in the query phase, we receive $A' \subseteq A$ as well as another set $C \subseteq \Real$ which we turn into a suitable query set $C' \subseteq \Field_p$ for the integer algorithm.
    For each $c \in C$, we start by binary searching the prepared sorted set of pairs to check if there is a pair $(a_i, a_j) \in A^2$ with $a_i + a_j = c$.
    If such a pair exists, we add $x_i + x_j \in \Field_p$ to $C'$.
    Otherwise, $c$ cannot be part of a solution, regardless of $A'$, so we can safely discard it.
    Finally, we hand $C'$ to the integer algorithm and return the query output.
    The correctness of the algorithm immediately follows from the following claim since $a_k + a_\ell = c$ for some $a_k, a_\ell \in A'$ is then equivalent to $a_k + a_\ell = a_i + a_j$.

    \begin{claim}
        With high probability, $a_i + a_j = a_k + a_\ell$ if and only if $x_i + x_j = x_k + x_\ell$ for all $i, j, k, \ell \in [n]$.
    \end{claim}
    \begin{proof}
        Apply the Frank--Tardos theorem (\Cref{thm:ft}) on $a = (a_1, \dots, a_n) \in \Real^n$ to obtain integers $v \in \Int^n$ bounded by $\norm{v}_\infty \leq \exp(O(n^3))$ such that $a_i + a_j = a_k + a_\ell$ if and only if $v_i + v_j = v_k + v_\ell$. We claim that, with high probability, this is further equivalent to $v_i + v_j \equiv v_k + v_\ell \mod{\,p}$. Indeed, the only failure event is when $v_i + v_j \neq v_k + v_\ell$ and $p$ divides $v_i + v_j - v_k - v_\ell$. This number has at most $\log (2 \norm{v}_\infty) \leq O(n^3)$ prime divisors, thus the probability of selecting $p \in [n^{10}, 2n^{10}]$ among these $O(n^3)$ bad choices is at most $\tilde O(n^{-7})$ by the Prime Number Theorem. By a union bound over the $n^4$ choices of $(i, j, k, \ell)$, the total error probability is $\tilde O(n^{-3})$.
        
        Thus, view $v$ as an element of $\Field_p^n$. By construction we have $M v = 0$ (now viewed over $\Field_p$). Sampling a uniformly random kernel element $x \in \ker M$ can therefore be equivalently viewed as sampling $y \in \ker M$ and a uniformly random scalar $\alpha \in \Field_p$, and setting $x \coloneqq y + \alpha v$. But then, for any quadruple $(i, j, k, \ell)$ with $a_i + a_j \neq a_k + a_\ell$ we have that $x_i + x_j - x_k - x_\ell = (y_i + y_j - y_k - y_\ell) + \alpha (v_i + v_j - v_k - v_\ell)$. Since $v_i + v_j - v_k - v_\ell \neq 0$ and since $\alpha \in \Field_p$ is sampled uniformly and independently, it follows that also $x_i + x_j - x_k - x_\ell$ is uniformly distributed in $\Field_p$. The probability of hitting zero is thus $1/p$, and by a union bound over all non-solution quadruples the total error probability is $O(n^{-6})$. 
    \end{proof}
    
    We finally analyze the running time. Preparing the sorted set of pairs as well as $S$ takes time $\tilde O(n^2)$. The call to \cref{thm:sparse-rect-eq-sys} is with parameters $m \le n^2 - 1$ and $t = 4$, and thus takes time $\tilde O(n^2)$, too. Finally, we preprocess one instance of Preprocessed 3SUM with Unknown $C$ over $\Field_p$ (or two integer instances) in time $O(P(n))$. In the query phase, we perform $n$ binary searches over a set of size $n^2$ and then issue two queries to respective two integer instances taking time $O(Q(n))$. The total query time is therefore $O(Q(n) + n\log n) = \tilde O(Q(n))$, as claimed.
\end{proof}

\section*{AI Disclosure}\label{sec:ai}
In an early stage of this project, the authors asked ChatGPT 5.6 Sol for a weaker but more efficient real-to-integer hash function than provided by the Frank--Tardos theorem that can exactly maintain 3SUM solutions.
The LLM suggested the general approach based on linear algebra outlined below \Cref{obs:ai} without running time analysis.
The authors then developed the remaining tools based on sparse linear algebra and the full algorithm without AI tools.

Finally, AI tools (ChatGPT 6 Astra) have been used to assist in drafting parts of the introduction and review the technical sections. The authors assume full responsibility.

\bibliographystyle{alphaurl}
\bibliography{main}

\appendix
\section{Preprocessed 3SUM with Subquadratic Preprocessing Time}
\label{sec:subquadratic}
In this section, we demonstrate how the recent subquadratic-time algorithm for 3SUM~\cite{AlmanV26} can be turned into an algorithm for Preprocessed 3SUM with subquadratic preprocessing time and $\tilde O(n^{3/2})$ query time:

\begin{theorem} \label{thm:preprocessed-subquadratic}
Preprocessed 3SUM with integer inputs (over the word RAM model) can be solved with preprocessing time $O(n^{2-\epsilon})$ and query time $\tilde O(n^{1.5})$, for some absolute constant $\epsilon > 0$.
\end{theorem}

We have made no attempt to optimize the constants in this section. Besides~\cite{AlmanV26}, the proof of \cref{thm:preprocessed-subquadratic} requires ideas from two further sources: the state-of-the-art algorithm for integer Preprocessed 3SUM~\cite{KasliwalPS25}, and a line of work on algorithms for structured 3SUM instances~\cite{ChanL15,AbboudBF23,JinX23,FischerJX25} based on the Balog--Szemerédi--Gowers (BSG) theorem from additive combinatorics~\cite{BalogS94,Gowers01} and sparse convolutions~\cite{ColeH02,BringmannFN22}. We start by gathering the required tools.

The following theorem states an algorithmic version of the BSG theorem; see e.g.~\cite[Appendix B]{AbboudBF23} for a proof. Here, for an integer set $A$, the \emph{additive energy} is $E(A) := \abs{\set{(a,b,c,d) \in A^4 : a+b=c+d}}$.

\begin{theorem}[Algorithmic BSG Theorem] \label{thm:bsg}
Let $A$ be an integer set with $E(A) \geq |A|^3 / K$. Then there is a subset $A' \subseteq A$ such that:
\begin{enumerate}[label=(\roman*)]
    \item $|A'| \geq \Omega(|A| / K^2)$, and
    \item $|A' + A'| \leq O(K^{26} |A|)$.
\end{enumerate}
The set $A'$ can be computed in time $\tilde O(K^{12} |A|)$ by a randomized algorithm.
\end{theorem}

\begin{corollary} \label{cor:bsg}
Let $A$ be an integer set, and let $k \geq 1$ be an integer. There is a partition $A = A_1 \sqcup \dots \sqcup A_k \sqcup R$ such that:
\begin{enumerate}[label=(\roman*)]
    \item $|A_i + A_i| \leq \tilde O(|A_i| k^{52})$,
    \item $\abs{\set{(a, b, c) \in R \times A \times A : a + b = c}} \leq \tilde O(|A|^2 / k^{1/8})$.
\end{enumerate}
The partition can be computed in time $\tilde O(k^7 |A|)$ by a randomized algorithm.
\end{corollary}
\begin{proof}
We keep a working set $A^*$ that is initially $A^* := A$. In each step we call \cref{thm:bsg} on $A^*$ (with some parameter $K$ to be determined) to extract some structured subset $A_i$, remove $A_i$ from $A^*$ and continue. Once $E(A^*) \leq |A|^3 / K$, we put $R := A^*$ and stop. In each step we remove at least an $\Omega(K^{-2})$-fraction of the elements of $A^*$; thus, after $O(K^2 \log |A|)$ steps the algorithm is forced to stop. For this reason, we may pick $K \coloneqq \tilde \Theta(\sqrt{k})$. 

Property~(i) is clear by construction. For Property~(ii), a simple application of the Cauchy--Schwarz inequality suffices to bound the number of solutions $(a, b, c) \in R \times A \times A$ to $a + b = c$ by at most $E(R)^{1/4} |A|^{5/4} \leq \tilde O(|A|^2 / k^{1/8})$.
The running time follows from the running time of \cref{thm:bsg} and the fact that we can approximately compute the energy using few samples by a Chernoff bound.
\end{proof}

We also rely on the following basic inequality from additive combinatorics; see, e.g.,~\cite[Corollary~6.29]{TaoV06}.

\begin{lemma}[Plünnecke--Ruzsa] \label{lem:plunnecke-ruzsa}
Let $A \subseteq \Int$. If $|A + A| \leq K |A|$, then $|n A - m A| \leq K^{n+m} |A|$ for all integers $n, m \geq 0$.
\end{lemma}

And finally we rely on a subroutine to efficiently compute structured sumsets. The following lemma is immediate from~\cite{BringmannFN22}.

\begin{lemma} \label{lem:sparse-conv}
There is a deterministic algorithm that, given sets $A, B, C \subseteq \set{-U, \dots, U}$, counts the number of 3SUM solutions $a + b = c$ with $a \in A, b \in B, c \in C$ in time $(|A + B| + |C|) (\log U)^{O(1)}$.
\end{lemma}

\begin{proof}[Proof of \cref{thm:preprocessed-subquadratic}]
The algorithm combines ideas from three different lines of work. First, suppose that the given set~$A$ contains at most $O(n^{2-\delta})$ 3SUM solutions for some constant $\delta$. We will later reduce to this case. 

The first idea is to follow the ideas of Kasliwal, Polak and Sharma~\cite{KasliwalPS25}. In the preprocessing phase we sample a uniformly random prime $p$ in the interval $[n^{3/2}, 2n^{3/2}]$. A standard argument shows that any integer triple $(a, b, c) \in A^3$ with $a + b \neq c$ satisfies $a + b \equiv c \mod p$ with probability at most $\tilde O(n^{-1.5})$. Let us call such a triple a \emph{false positive}. Then by Markov's bound, with constant probability, the total number of false positives is at most $\tilde O(n^{1.5})$. The idea is to list \emph{all} false positives in the preprocessing stage; we will describe how to implement this step efficiently in a minute. Upon query of a subset $A' \subseteq A$ we will then \emph{count} all solutions $(a, b, c) \in (A')^3$ with $a + b \equiv c \mod p$. This can be implemented by FFT in time $\tilde O(p) = \tilde O(n^{1.5})$. We then enumerate all false positives $(a, b, c)$, and subtract $1$ whenever $a, b, c \in A'$. What remains is the true number of 3SUM solutions in $A'$. We report YES if and only if that count is positive.

Second, we describe how to efficiently list all false positives. This is based on the recent $O(n^{1.9999})$-time algorithm for 3SUM~\cite{AlmanV26}. A well-known result in fine-grained complexity is that if 3SUM can be solved in time $O(n^{2-\epsilon})$, then there is also an algorithm to \emph{list} up to $O(n^{2-\delta})$ 3SUM solutions in time $\tilde O(n^{2-\epsilon\delta/2})$ (see, e.g., \cite{VassilevskaW18,DurajKPW20,FischerPS26}). In particular, we can afford to list all solutions $(a, b, c) \in A^3$ with $a + b \equiv c \mod p$. These include the at most~$O(n^{2-\delta})$ true solutions (by our initial assumption) plus at most $\tilde O(n^{1.5})$ false positives. Afterwards, we scan over the whole list and only retain the false positives. The total time is truly subquadratic, $O(n^{2-0.00005\delta})$.

Third, it remains to remove the assumption that the number of 3SUM solutions in $A$ is at most $O(n^{2-\delta})$. Let $k = \ceil{n^{\gamma}}$ and $t = \ceil{n^\beta}$ be parameters with small $\beta, \gamma > 0$ to be determined. We first partition the set $A$ by \cref{cor:bsg} into structured parts $A_1, \dots, A_k$ and an unstructured remainder $R$. For each set $A_i$, we further define a partition $A = H_i \sqcup L_i$ so that
\begin{itemize}
    \item for each $b \in H_i$, $\abs{\set{(a, c) \in A_i \times A : a + b = c}} \geq n / t$, and
    \item for each $b \in L_i$, $\abs{\set{(a, c) \in A_i \times A : a + b = c}} \leq 2n / t$.
\end{itemize}
These partitions can easily be computed in total time $\tilde O(k t n) = \tilde O(n^{1+\gamma+\beta})$ by sampling (see~\cite{FischerJX25} for a deterministic implementation).

We now follow the same approach as before with the exception that we will not list all false positives. Specifically, we distinguish three types of modular solutions $(a, b, c) \in A^3$ with $a + b \equiv c \mod p$:
\begin{enumerate}
    \item $a \in R$,
    \item $a \in A_i$ and $b \in L_i$,
    \item $a \in A_i$ and $b \in H_i$.
\end{enumerate}
\cref{cor:bsg} implies that there are at most $\tilde O(n^2 / k^{1/8}) + \tilde O(n^{1.5}) = \tilde O(n^{2-\gamma/8})$ modular solutions of type~(1). We can afford to list these exactly as before in time $O(n^{2-0.000005\gamma})$. Moreover, for each $i \in [k]$, there are at most $n \cdot 2n / t + \tilde O(n^{1.5}) = O(n^2 / t) + \tilde O(n^{1.5}) = O(n^{2-\beta})$ modular solutions of type~(2), so these can be listed in subquadratic preprocessing time $O(n^{2-0.00005\beta})$ per $i$, and therefore time $O(n^{2+\gamma-0.00005\beta})$ in total.

This only leaves the false positives of type~(3). These we cannot list explicitly. Instead, we will compute the number of false positives of type~(3) only in the query phase. That is, given $A' \subseteq A$, the goal is to compute
\begin{equation*}
    F_i = \abs{\set{(a, b, c) \in (A' \cap A_i) \times (A' \cap H_i) \times A' : a + b \neq c, a + b \equiv c \mod p}}.
\end{equation*}
In fact, we will compute $F_i$ as the difference $M_i - T_i$ defined by
\begin{align*}
    M_i &= \abs{\set{(a, b, c) \in (A' \cap A_i) \times (A' \cap H_i) \times A' : a + b \equiv c \mod p}}, \\
    T_i &= \abs{\set{(a, b, c) \in (A' \cap A_i) \times (A' \cap H_i) \times A' : a + b = c}}.
\end{align*}
The quantity $T_i$ can be evaluated by one call to \cref{lem:sparse-conv} in time $\tilde O(|A_i + H_i| + n)$. The quantity $M_i$ can be evaluated similarly, along with the usual trick to count modular solutions by computing integer solutions for the instance with and without overflow. The running time is also $\tilde O(|A_i + H_i| + n)$.

To bound the total running time, suppose (only for the sake of the analysis) that a set $C \subseteq A$ is sampled uniformly at random with rate $t / n \cdot 100 \log n$. This rate is chosen such that, with high probability, each $b \in H_i$ is involved in at least one solution $a + b = c$ with $a \in A_i$ and $c \in C$. In other words, $H_i \subseteq C - A_i$. Thus, we can bound the sumset size by
\begin{equation*}
    |A_i + H_i| \leq |A_i - A_i + C| \leq |A_i - A_i| \, |C|.
\end{equation*}
To further bound $|A_i - A_i|$ we apply the Plünnecke--Ruzsa inequality (\cref{lem:plunnecke-ruzsa}) and the fact that $A_i$ has doubling at most $\tilde O(k^{52})$. Moreover, the expected size of $C$ is $\tilde O(t)$. Thus,
\begin{equation*}
    |A_i + H_i| \leq \tilde O(k^{104} t n),
\end{equation*}
and the total time to compute the number of false positives of type~(3) can thus be bounded by
\begin{equation*}
    \tilde O\parens*{k n + \sum_{i=1}^k |A_i + H_i|} = \tilde O(k^{105} t n) = \tilde O(n^{1+105\gamma+\beta}).
\end{equation*}
All in all, we have worsened the preprocessing time by computing the partitions in time $\tilde O(n k^7 + n k t)$ and listing the modular solutions of types~(1) and~(2) in time $O(n^{2-0.000005\gamma} + n^{2+\gamma-0.00005\beta})$, and we have worsened the query time by $\tilde O(n^{1+105\gamma+\beta})$. For example, choosing $\beta=1/4$ and $\gamma=10^{-6}$ keeps the preprocessing time truly subquadratic and the query time $\tilde O(n^{1.5})$.
\end{proof}

\end{document}